\documentclass[12pt]{article}
\usepackage{enumerate}
\usepackage{amsfonts}
\usepackage{amsmath}
\usepackage{amssymb}
\usepackage{amsthm}
\usepackage{latexsym}
\usepackage{color}

\def\H{\mathcal{H}}

\def\P{\mathcal{P}}

\def\S{\mathfrak{S}}

\def\F{\mathfrak{F}}
\def\C{\mathfrak{C}}
\def\T{\mathfrak{T}}

\newcommand{\supp}{\mathrm{supp}}
\newcommand{\rank}{\mathrm{rank}}

\newcommand{\Tr}{\mathrm{Tr}}

\newcommand{\shs}{\hspace{1pt}}
\newcounter{defin}  \newcounter{lemma}  \newcounter{theorem}
\newcounter{proposition} \newcounter{corol}  \newcounter{remark} \newcounter{example}
\newenvironment{lemma}{\par\refstepcounter{lemma}     \textbf{Lemma \thelemma.} }{\rm\par}

\newenvironment{proposition}{\par\refstepcounter{proposition}     \textbf{Proposition \theproposition.}\ }{\rm\par}

\newenvironment{remark}{\par\refstepcounter{remark}     \textbf{Remark \theremark.}}{\rm\par}
\newenvironment{example}{\par\refstepcounter{example}     \textbf{Example \theexample.}}{\rm\par}

\begin{document}

\title{How to improve the semicontinuity bounds in [Lett. Math. Phys., 113, 121 (2023)]}

\author{M.E.~Shirokov\footnote{email:msh@mi.ras.ru}\\
Steklov Mathematical Institute, Moscow, Russia}
\date{}
\maketitle
\begin{abstract}
We show how to  improve the semicontinuity bounds  in \cite{LCB} by optimizing
the proof of the basic technical lemma. In this optimization we apply the modified version of the trick used in the resent
article \cite{D++}.

The most important applications are the semicontinuity bound for the von Neumann entropy  with the energy constraint
and the semicontinuity bounds for the entanglement of formation  with the rank/energy constraint.
\end{abstract}

\section{Notation and necessary facts}\label{sec2}

Let $\mathcal{H}$ be a separable Hilbert space,
$\mathfrak{B}(\mathcal{H})$ the algebra of all bounded operators on $\mathcal{H}$ with the operator norm $\|\cdot\|$ and $\mathfrak{T}( \mathcal{H})$ the
Banach space of all trace-class
operators on $\mathcal{H}$  with the trace norm $\|\!\cdot\!\|_1$. Let
$\mathfrak{S}(\mathcal{H})$ be  the set of quantum states (positive operators
in $\mathfrak{T}(\mathcal{H})$ with unit trace) \cite{H-SCI,N&Ch,Wilde}.

Write $I_{\mathcal{H}}$ for the unit operator on a Hilbert space
$\mathcal{H}$.

Let $H$ be a positive (semi-definite)  operator on a Hilbert space $\mathcal{H}$.  For any positive operator $\rho\in\T(\H)$ we will define the quantity $\Tr H\rho$ by the rule
\begin{equation*}
\Tr H\rho=
\left\{\begin{array}{l}
        \sup_n \Tr P_n H\rho\;\; \textrm{if}\;\;  \supp\rho\subseteq {\rm cl}(\mathcal{D}(H))\\
        +\infty\;\;\textrm{otherwise}
        \end{array}\right.
\end{equation*}
where $P_n$ is the spectral projector of $H$ corresponding to the interval $[0,n]$ and ${\rm cl}(\mathcal{D}(H))$ is the closure of the domain of $H$. If
$H$ is the Hamiltonian (energy observable) of a quantum system described by the space $\H$ then
$\Tr H\rho$ is the mean energy of a state $\rho$.

We will say that a positive operator $H$ satisfies  the \emph{Gibbs condition} if
\begin{equation}\label{H-cond}
  \Tr\, e^{-\beta H}<+\infty\quad\textrm{for all}\;\,\beta>0.
\end{equation}
If this condition holds then the von Neumann entropy is continuous on the set
$$
\C_{H,E}=\left\{\rho\in\S(\H)\,\vert\,\Tr H\rho\leq E\right\}
$$
for any $E>0$ and attains the maximal value on this set at the \emph{Gibbs state}
\begin{equation}\label{Gibbs}
\gamma_H(E)\doteq e^{-\beta(E) H}/\Tr e^{-\beta(E) H},
\end{equation}
where the parameter $\beta(E)$ is determined by the equation $\Tr H e^{-\beta H}=E\Tr e^{-\beta H}$ \cite{W}.

We will use the function
\begin{equation}\label{F-def}
F_{H}(E)\doteq\sup_{\rho\in\C_{H,E}}S(\rho)=S(\gamma_H(E)).
\end{equation}
By Proposition 1 in \cite{EC} the Gibbs condition (\ref{H-cond}) is equivalent to the following asymptotic property
\begin{equation*}
  F_{H}(E)=o\shs(E)\quad\textrm{as}\quad E\rightarrow+\infty.
\end{equation*}

We will often assume that
\begin{equation}\label{star}
  E_0\doteq\inf\limits_{\|\varphi\|=1}\langle\varphi\vert H\vert\varphi\rangle=0.
\end{equation}

An important role further is plaid by the \emph{binary entropy} 
\begin{equation}\label{be}
  h_2(x)\doteq -x\ln x-(1-x)\ln (1-x),\quad x\in[0,1]. 
\end{equation}

\section{The Alicki-Fannes-Winter method in the quasi-classical settings: advanced version and its use}\label{sec3}

\subsection{New basic lemma}

In this subsection we describe a general result concerning properties of a function $f$ on a convex subset $\S_0$ of $\S(\H)$ taking values in $(-\infty,+\infty]$ and satisfying the inequalities
\begin{equation}\label{LAA-1}
  f(p\rho+(1-p)\sigma)\geq pf(\rho)+(1-p)f(\sigma)-a_f(p)
\end{equation}
and
\begin{equation}\label{LAA-2}
  f(p\rho+(1-p)\sigma)\leq pf(\rho)+(1-p)f(\sigma)+b_f(p),
\end{equation}
for all states $\rho$ and $\sigma$ in $\S_0$ and any $p\in[0,1]$, where $a_f(p)$ and $b_f(p)$ are continuous  functions on $[0,1]$ such that $a_f(0)=b_f(0)=0$.
These  inequalities can be treated, respectively, as weakened forms of concavity and convexity. We will call functions
satisfying both inequalities (\ref{LAA-1}) and (\ref{LAA-2}) \emph{locally almost affine} (breifly, \emph{LAA functions}), since for any such function $f$ the quantity
$\,\vert f(p\rho+(1-p)\sigma)-p f(\rho)-(1-p)f(\sigma)\vert \,$ tends to zero as $\,p\rightarrow 0^+$ uniformly on $\,\S_0\times\S_0$.
\smallskip

Let $\{X,\F\}$ be a measurable space and $\tilde{\omega}(x)$ a $\F$-measurable $\S(\H)$-valued function on $X$. Denote by $\P(X)$ the set of all probability measures on $X$ (more precisely, on $\{X,\F\}$). We will assume that the function $\tilde{\omega}(x)$ is integrable (in the Pettis sense \cite{P-int}) w.r.t. any measure in $\P(X)$. Consider  the set of states
\begin{equation}\label{q-set}
\mathfrak{Q}_{X,\F,\tilde{\omega}}\doteq\left\{\rho\in\S(\H)\,\left\vert\,\exists\mu_{\rho}\in\P(X):\rho=\int_X\tilde{\omega}(x)\mu_{\rho}(dx)\;\right.\right\}.
\end{equation}
We will call any measure $\mu_{\rho}$ in $\P(X)$ such that $\rho=\int_X\tilde{\omega}(x)\mu_{\rho}(dx)$ a \emph{representing measure} for a state $\rho$ in $\mathfrak{Q}_{X,\F,\tilde{\omega}}$.

We will use the total variation distance between probability measures $\mu$ and $\nu$ in $\P(X)$ defined as
\begin{equation}\label{TVD-def}
\mathrm{TV}(\mu,\nu)=\sup_{A\in\F}\vert\mu(A)-\nu(A)\vert.
\end{equation}

Concrete  examples of sets $\mathfrak{Q}_{X,\F,\tilde{\omega}}$ can be found in \cite{LCB}.

The following lemma gives semi-continuity bounds for  LAA functions on a set of quantum states having form (\ref{q-set}). It is proved by obvious modification of the Alicki-Fannes-Winter technique.

\begin{lemma}\label{g-ob} \emph{Let $\mathfrak{Q}_{X,\F,\tilde{\omega}}$ be the set defined in (\ref{q-set}) and $\S_0$ a convex subset of $\S(\H)$ with the property
}\begin{equation}\label{S-prop}
  \rho\in\S_0\cap\mathfrak{Q}_{X,\F,\tilde{\omega}}\quad \Rightarrow \quad\{\sigma\in\mathfrak{Q}_{X,\F,\tilde{\omega}}\,\vert\,\exists \varepsilon>0:\varepsilon\sigma\leq \rho\}\subseteq\S_0.
\end{equation}

\emph{Let $f$ be a function  on the set $\,\S_0$ taking values in $(-\infty,+\infty]$ that satisfies inequalities (\ref{LAA-1}) and (\ref{LAA-2}) with possible
value $+\infty$ in both sides. Let $\rho$ and $\sigma$ be states in $\,\mathfrak{Q}_{X,\F,\shs\tilde{\omega}}\cap\S_0\,$ with representing measures $\mu_{\rho}$ and $\mu_{\sigma}$ correspondingly. If $\,f(\rho)<+\infty\,$ then
\begin{equation}\label{AFW-1+}
f(\rho)-f(\sigma)\leq \varepsilon C_f(\rho,\sigma\shs\vert\shs\varepsilon)+a_f(\varepsilon)+b_f(\varepsilon),
\end{equation}
where $\varepsilon=\mathrm{TV}(\mu_{\rho},\mu_{\sigma})$,
\begin{equation}\label{C-f}
C_f(\rho,\sigma\shs\vert\shs\varepsilon)\doteq\sup\left\{f(\varrho)-f(\varsigma)\left\vert\, \varrho,\varsigma\in\mathfrak{Q}_{X,\F,\shs\tilde{\omega}},\; \varepsilon\varrho\leq\rho,\; \varepsilon\varsigma\leq\sigma\right.\right\}
\end{equation}
and the left hand side of (\ref{AFW-1+}) may be equal to $-\infty$.}

\emph{If the function $f$ is nonnegative then inequality (\ref{AFW-1+}) holds with $C_f(\rho,\sigma\shs\vert\shs\varepsilon)$ replaced by}
\begin{equation}\label{C-f+}
C^+_f(\rho\shs\vert\shs\varepsilon)\doteq\sup\left\{f(\varrho)\,\left\vert\, \varrho\in\mathfrak{Q}_{X,\F,\tilde{\omega}},\; \varepsilon\varrho\leq\rho\right.\right\}.
\end{equation}
\end{lemma}

\begin{proof} We may assume that $f(\sigma)<+\infty$, since otherwise (\ref{AFW-1+}) holds trivially. By the condition we have
$$
2\mathrm{TV}(\mu_{\rho},\mu_{\sigma})=[\mu_{\rho}-\mu_{\sigma}]_+(X)+[\mu_{\rho}-\mu_{\sigma}]_-(X)=2\varepsilon,
$$
where $[\mu_{\rho}-\mu_{\sigma}]_+$ and $[\mu_{\rho}-\mu_{\sigma}]_-$ are the positive and negative parts
of the measure $\mu_{\rho}-\mu_{\sigma}$ (in the sense of Jordan decomposition theorem \cite{Bil}). Since $\,\mu_{\rho}(X)=\mu_{\sigma}(X)=1$, it follows from the above equality that
$\,[\mu_{\rho}-\mu_{\sigma}]_+(X)=[\mu_{\rho}-\mu_{\sigma}]_-(X)=\varepsilon$. Hence, $\nu_\pm\doteq \varepsilon^{-1}[\mu_{\rho}-\mu_{\sigma}]_\pm\in\P(X)$.
Moreover, it is easy to show, by using the definition of $\,[\mu_{\rho}-\mu_{\sigma}]_\pm$ via the Hahn decomposition of $X$, that
\begin{equation}\label{m-r}
\varepsilon\nu_+=[\mu_{\rho}-\mu_{\sigma}]_+\leq \mu_{\rho}\quad \textrm{and} \quad \varepsilon\nu_-=[\mu_{\rho}-\mu_{\sigma}]_-\leq \mu_{\sigma}.
\end{equation}

Modifying the idea used in \cite{D++} consider the states
\begin{equation}\label{tau-s}
\tau_+=\int_X\tilde{\omega}(x)\nu_+(dx),\quad\tau_-=\int_X\tilde{\omega}(x)\nu_-(dx)\quad \textrm{and} \quad \omega_*=\int_X\tilde{\omega}(x)\mu_*(dx),
\end{equation}
where
$$
\mu_*=\frac{\mu_{\rho}-\varepsilon\nu_+}{1-\varepsilon}=\frac{\mu_{\sigma}-\varepsilon\nu_-}{1-\varepsilon}
$$
is a measure in $\P(X)$. Since the inequalities in (\ref{m-r}) imply that $\varepsilon\tau_+\leq\rho$, $\varepsilon\tau_-\leq\sigma$ and $(1-\varepsilon)\omega_*\leq\rho$, these states
belong to the set $\mathfrak{Q}_{X,\F,\tilde{\omega}}\cap\S_0$ due to condition (\ref{S-prop}).

Then we have
\begin{equation}\label{omega-star}
\rho=\varepsilon\tau_++(1-\varepsilon)\omega_*\quad \textrm{and} \quad\sigma=\varepsilon\tau_-+(1-\varepsilon)\omega_*.
\end{equation}
By applying inequalities (\ref{LAA-1}) and (\ref{LAA-2}) to the decompositions  in (\ref{omega-star}) we  obtain
$$
f(\rho)\leq\varepsilon f(\tau_+)+(1-\varepsilon)f(\omega_*)+a_f(\varepsilon)\quad \textrm{and} \quad f(\sigma)\geq\varepsilon f(\tau_-)+(1-\varepsilon)f(\omega_*)-b_f(\varepsilon).
$$
The last inequality implies the finiteness of $f(\omega_*)$ by the assumed finiteness of $f(\sigma)$. So,
these inequalities show that
$$
f(\rho)-f(\sigma)\leq\varepsilon(f(\tau_+)-f(\tau_-))+a_f(\varepsilon)+b_f(\varepsilon).
$$
Since $\varepsilon\tau_+\leq\rho$ and $\varepsilon\tau_-\leq\sigma$, this implies inequality (\ref{AFW-1+}).

The last claim of the lemma is obvious.
\end{proof}


\begin{remark}\label{g-ob-r} In general, the condition $\mathrm{TV}(\mu_{\rho},\mu_{\sigma})=\varepsilon$ in
Lemma \ref{g-ob} can not be replaced by the condition $\mathrm{TV}(\mu_{\rho},\mu_{\sigma})\leq\varepsilon$, since
the functions $\varepsilon \mapsto C_f(\rho,\sigma\shs\vert\shs\varepsilon)$ and $\varepsilon \mapsto C^+_f(\rho\shs\vert\shs\varepsilon)$ may be decreasing and, hence, special arguments are required to show that the r.h.s. of (\ref{AFW-1+}) is a nondecreasing function of $\varepsilon$.
\end{remark}\smallskip

\begin{remark}\label{S-prop-r} Condition (\ref{S-prop}) in Lemma \ref{g-ob} can be replaced by the condition
\begin{equation*}
  \rho,\sigma\in\S_0\cap\mathfrak{Q}_{X,\F,\tilde{\omega}}\quad \Rightarrow \quad \tau_+,\tau_-,\omega_*\in\S_0,
\end{equation*}
where $\tau_+$, $\tau_-$ and $\omega_*$ are the states defined in (\ref{tau-s}) via the representing measures $\mu_{\rho}$, $\mu_{\sigma}$ and $\mu_{*}$. In this case
one should correct the definitions of $C_f(\rho,\sigma\shs\vert\shs\varepsilon)$ and $C^+_f(\rho\shs\vert\shs\varepsilon)$
by replacing $\mathfrak{Q}_{X,\F,\shs\tilde{\omega}}$ in (\ref{C-f}) and (\ref{C-f+}) with $\mathfrak{Q}_{X,\F,\shs\tilde{\omega}}\cap\S_0$.
\end{remark}

\subsection{On advantage that new Lemma 1 gives}

The advantage of Lemma \ref{g-ob} in the previous subsection in comparison with Lemma 1 in Section 3.1 in \cite{LCB}
consists in replacing the term
$$D_f(\varepsilon)\doteq\displaystyle(1+\varepsilon)(a_f+b_f)\!\left(\frac{\varepsilon}{1+\varepsilon}\right)$$
in both claims with the strictly smaller term $$\,a_f(\varepsilon)+b_f(\varepsilon).$$

Since in all the applications of Lemma 1 in Section 3.1 in \cite{LCB} considered in Sections 3.2 and 4 in \cite{LCB} we deal with a function
$f$  satisfying inequalities (\ref{LAA-1}) and (\ref{LAA-2}) in which
$a_f$ and $b_f$ are functions proportional to the binary entropy $h_2$ defined in (\ref{be}), the use of the "advanced" Lemma \ref{g-ob}
in the previous subsection instead of Lemma 1 in Section 3.1 in \cite{LCB} in all the proofs leads to
the replacement
\begin{equation}\label{r}
g(\varepsilon)\doteq\displaystyle(1+\varepsilon)h_2\left(\frac{\varepsilon}{1+\varepsilon}\right)\quad \rightarrow \quad \tilde{h}_2(\varepsilon)\doteq\left\{\begin{array}{l}
        h_2(\varepsilon)\;\; \textrm{if}\;\;  \varepsilon\in[0,\frac{1}{2}]\\
        \ln2\quad \textrm{if}\;\;\;  \varepsilon\in[\frac{1}{2},1]
        \end{array}\right.
\end{equation}
in the right hand sides of all the  semicontinuity bounds obtained in \cite{LCB}.\footnote{We have to use $\tilde{h}_2$ instead of $h_2$ because the monotonicity of $g$ is exploited in the proofs of Theorems 1 and 2 in \cite{LCB}.}\medskip

So, we come to the following \medskip
\begin{proposition}\label{main}\emph{The inequalities in Theorems 1,2 and in Corollaries 1,2 in Section 3.2 in \cite{LCB} can be improved by applying the replacement (\ref{r}) in
their right hand sides.   }\medskip

\emph{The same improvement can be done in all the propositions, corollaries and examples presented in Section 4 in \cite{LCB}.}\medskip

\emph{More formally, the optimizing replacement (\ref{r}) can be done in  the right hand sides of the following inequalities from \cite{LCB}:
(31), (32), (37), (38), the inequality in Corollary 1A, (43), the inequality in Corollary 2A, (44), (45), (46), the inequality after (46), (47), (48), (52), (54), the inequality in Corollary 3, the inequality at the end of Section 4.1, (62), (63), (64), (65), (68), (69), (71), (72), (76), (77), (81), (82), (83), (84), (85), (87), the inequality before (91), (91), (92), (93).}

\end{proposition}

\medskip

\textbf{Note:} The term $\varepsilon g(E/\varepsilon)$ in the r.h.s. of (87),(92) and (93) can not been changed.
\medskip

\begin{example}\label{CB-1} According to Proposition \ref{main} the semicontinuity bound for the von Neumann entropy $S$
presented in Proposition 1 in \cite{LCB} can be improved as follows:\medskip

\emph{Let $H$ be a positive operator on $\H$  satisfying conditions (\ref{H-cond}) and (\ref{star}). If $\rho$ is a state in $\S(\H)$ such that $\Tr H\rho\leq E$ then
\begin{equation}\label{W-CB-2}
   S(\rho)-S(\sigma)\leq \varepsilon F_H((E-E_{H,\shs\varepsilon}(\rho))/\varepsilon)+\tilde{h}_2(\varepsilon)\leq \varepsilon F_H(E/\varepsilon)+\tilde{h}_2(\varepsilon)
\end{equation}
for any state $\sigma$ in $\S(\H)$ such that $\,\frac{1}{2}\|\rho-\sigma\|_1\leq \varepsilon$, where $E_{H,\shs\varepsilon}(\rho)\doteq \Tr H[\rho-\varepsilon I_{\H}]_+$ and the l.h.s. of (\ref{W-CB-2}) may be equal to $-\infty$.}\footnote{$\,[\rho-\varepsilon I_{\H}]_+$ is the positive part of the Hermitian operator $\,\rho-\varepsilon  I_{\H}$.}
\end{example}

\medskip\pagebreak

\begin{example}\label{CB-2} According to Proposition \ref{main} the semicontinuity bounds for the entanglement of formation (EoF)
presented in Proposition 4 in \cite{LCB} can be improved as follows:\footnote{$E_F^d$ and $E_F^c$ are discrete and continuous versions the EoF defined, respectively, by the expressions
\begin{equation}\label{E_F-def-d}
E_F^d(\omega)=\!\inf_{\sum_k\!p_k\omega_k=\omega}\sum_kp_kS([\omega_k]_A),\quad
\end{equation}
\begin{equation}\label{E_F-def-c}
E_F^c(\omega)=\!\inf_{\int\omega'\mu(d\omega')=\omega}\int\! S(\omega'_A)\mu(d\omega'),
\end{equation}
where the  infimum in (\ref{E_F-def-d}) is over all countable ensembles $\{p_k, \omega_k\}$ of pure states in $\S(\H_{AB})$ with the average state $\omega$ and the  infimum in (\ref{E_F-def-c}) is over all Borel probability measures on the set of pure states in $\S(\H_{AB})$ with the barycenter $\omega$.}\medskip

\emph{Assime that $AB$ is an infinite-dimensional bipartite quantum system.}\smallskip

\emph{\noindent \emph{A)} If $\rho$ is a state in $\S(\H_{AB})$ such that $\rank\rho_A$ is finite then
\begin{equation}\label{EF-CB-A}
   E^*_F(\rho)-E^*_F(\sigma)\leq \delta\ln(\rank\rho_A)+\tilde{h}_2(\delta),\quad E^*_F=E_F^d,E_F^c,
\end{equation}
for any state $\sigma$ in $\S(\H_{AB})$ such that $\,\frac{1}{2}\|\rho-\sigma\|_1\leq \varepsilon\leq1$,  where $\delta=\sqrt{\varepsilon(2-\varepsilon)}$ and  the l.h.s. of (\ref{EF-CB-A}) may be equal to $-\infty$.}\smallskip

\emph{\noindent \emph{B)} If $\rho$ is a state in $\S(\H_{AB})$ such that $\Tr H\rho_A\leq E$, where $H$ is a positive operator on $\H_A$  satisfying conditions (\ref{H-cond}) and (\ref{star}), then
\begin{equation}\label{EF-CB-B}
   E^*_F(\rho)-E^*_F(\sigma)\leq \delta F_H(E/\delta)+\tilde{h}_2(\delta),\quad E^*_F=E_F^d,E_F^c,
\end{equation}
for any state $\sigma$ in $\S(\H_{AB})$ such that $\,\frac{1}{2}\|\rho-\sigma\|_1\leq \varepsilon\leq1$, where $\,\delta=\sqrt{\varepsilon(2-\varepsilon)}$ and  the l.h.s. of  (\ref{EF-CB-B}) may be equal to $-\infty$.}
\end{example}\smallskip

\textbf{Note:} An improved version of  semicontinuity bound (\ref{EF-CB-A}) is presented in \cite{AMS}.
\smallskip

\begin{example}\label{CB-3} According to Proposition \ref{main} the semicontinuity bound for the Shannon conditional entropy (equivocation) $H(X_1\vert X_2)$
presented in \cite[inequality (87)]{LCB} can be improved as follows:

Assume that $(X_1,X_2)$ and $(Y_1,Y_2)$ are pairs of discrete random variables
and that the random variables $X_1$ and $Y_1$ take the values $0,1,2,..$. Then
\begin{equation}\label{eq-2-cb}
    H(X_1\vert X_2)_{\bar{p}}-H(Y_1\vert Y_2)_{\bar{q}}\leq \varepsilon g(E/\varepsilon)+\tilde{h}_2(\varepsilon)
\end{equation}
for any 2-variate probability  distributions $\bar{p}=\{p_{ij}\} $ and $\bar{q}=\{q_{ij}\}$ (describing the pairs $(X_1,X_2)$ and $(Y_1,Y_2)$) such that  $\mathbb{E}(X_1)\doteq\sum_{i,j=1}^{+\infty} (i-1)\shs p_{ij}\leq E$ and $\mathrm{TV}(\bar{p},\bar{q})\leq \varepsilon$.\smallskip

According to the remark after Proposition \ref{main} the term $g(E/\varepsilon)$ in (\ref{eq-2-cb}) is not changed. 

\smallskip

The example presented after  inequality (87) in \cite{LCB} shows that the semicontinuity bound (\ref{eq-2-cb}) is close-to-optimal.
\end{example}

\end{document}